\documentclass{article}
\usepackage{graphicx}
\usepackage{amsfonts}
\usepackage{amsmath}
\usepackage{amssymb}
\usepackage{amsthm}
\usepackage[margin=0.95in]{geometry}
\usepackage{tcolorbox}
\usepackage{url}
\usepackage{hyperref}
\usepackage{multirow}
\usepackage{booktabs}
\usepackage{enumitem}
\usepackage[normalem]{ulem}
\usepackage{natbib} 
\usepackage{authblk}
\DeclareMathOperator*{\argmax}{arg\,max}

\title{Coalitions on the Fly in Cooperative Games}

\author[1]{Yao Zhang\thanks{zhang@agent.inf.kyushu-u.ac.jp}}
\author[2]{Indrajit Saha\thanks{indrajit@inf.kyushu-u.ac.jp}}
\author[3]{Zhaohong Sun\thanks{sunzhaohong1991@gmail.com}}
\author[4]{Makoto Yokoo\thanks{yokoo@inf.kyushu-u.ac.jp}}
\affil[1,2,3,4]{Kyushu University, Japan}

\date{}

\newtheorem{definition}{Definition}

\newtheorem{proposition}{Proposition}

\newtheorem{example}{Example}
\newtheorem{theorem}{Theorem}
\newtheorem{lemma}{Lemma}
\newtheorem{assumption}{Assumption}

\usepackage{lineno}

\begin{document}

\maketitle

\begin{abstract}
In this work, we examine a sequential setting of a cooperative game in which players arrive dynamically to form coalitions and complete tasks either together or individually, depending on the value created. Upon arrival, a new player as a decision maker faces two options: forming a new coalition or joining an existing one. We assume that players are greedy, i.e., they aim to maximize their rewards based on the information available at their arrival. The objective is to design an online value distribution policy that incentivizes players to form a coalition structure that maximizes social welfare. We focus on monotone and bounded cooperative games. Our main result establishes an upper bound of $\frac{3\mathsf{min}}{\mathsf{max}}$ on the competitive ratio for any irrevocable policy (i.e., one without redistribution), and proposes a policy that achieves a near-optimal competitive ratio of $\min\left\{\frac{1}{2}, \frac{3\mathsf{min}}{\mathsf{max}}\right\}$, where $\mathsf{min}$ and $\mathsf{max}$ denote the smallest and largest marginal contribution of any sub-coalition of players respectively. Finally, we also consider non-irrevocable policies, with alternative bounds only when the number of players is limited.
\end{abstract}

\section{Introduction}
Cooperative game theory, a classic area 
in game theory, considers how to distribute the value created by a coalition of players~\citep{driessen2013cooperative, shapley1953value, von2007theory}. In real-world applications, two major concerns have been discussed in the literature. The first is that players may not be willing to form a grand coalition, even if every player brings a positive marginal contribution to the others. This can occur because additional players may play symmetric roles without generating enough extra earnings~\citep{aumann1974cooperative}. A key concept addressing this issue is the \emph{coalition structure}, which allows players to form sub-coalitions 
{next to} a grand coalition~\citep{hart1983endogenous, ueda2018coalition}. The key challenges then include 
how to partition the players into a coalition structure and how to distribute the value generated by the coalitions while aiming for higher social welfare and satisfying properties such as stability, i.e., a player will not be attracted to another coalition. The second concern is that, in many scenarios, not all players are able to join the game at the beginning. That is, players may require invitations to join~\citep{zhang2022incentives}, or they may arrive at different times, \emph{in an online manner}~\citep{aziz2025participation, ge2024incentives}, so that players' strategic behaviors must also be taken into account.

In this paper, we consider scenarios where both concerns arise simultaneously—namely, the formation of coalition structures  in an online manner. {There are many real-world applications where it is realistic to assume that agents arrive over time and the entire input is not available from the beginning;} for example, consider a regional economic plan initiated by the government, which aims to attract investment or launch projects across different markets or industries. In such a setting, people may arrive at different times and can choose to invest in or join an existing project, or to start a new one. Similarly, in the case of a large-scale social or experimental project involving multiple regions, new participants may join an already-formed group or establish a new local group as part of the expansion through recruitment, {or people entering social networks. Other possible applications can also be found in \citep{flammini2021online}.}

All these scenarios suggest an online cooperative game with coalition structures. As with most cooperative games, we need to design a distribution policy to allocate the value earned by the players. Since players arrive sequentially and unpredictably, such a policy must be implementable in an online manner—this is also where the main challenge lies. The primary objective is to guide players in forming a coalition structure that maximizes social welfare, i.e., the total value created by all coalitions. Regarding players' behaviors, we assume they have no  prior knowledge of the future and must act myopically as greedy agents. More precisely, the key research question is: \emph{How to design an online policy that guides sequentially arriving greedy players to form coalitions that achieve near-optimal social welfare?}

\noindent\textbf{Contributions.} Within this setting, 
our key contributions as follows.
\begin{itemize}
    \item We introduce a model of online cooperative games with coalition structures, which considers how coalition structures are formed and how to guide players to form appropriate coalitions through policy design (Section~\ref{sec:model}).
    \item In the domain of monotone and bounded cooperative games, 
    involving greedy players, we provide upper and lower bounds on the worst-case competitive ratio for social welfare under irrevocable policies (where reallocation is not allowed). The lower bound is derived by allocating marginal contributions exceeding a threshold, which is very close to the proven upper bound (Section~\ref{sec:irrevocable}).
    \item We also model online policies that are non-irrevocable. Here, we additionally consider a variant of greedy players and provide both upper and lower bounds of the competitive ratio in each case, as shown in Table~\ref{tab:results}. Especially, without additional assumption indicated by (*), it shares the same bounds as above, which suggests that even when decisions are irrevocable, it is still possible to design a policy that achieves near-optimal social welfare, with a provable competitive ratio bound in general cases (Section~\ref{sec:non-irr}).
\end{itemize}

{\noindent\textbf{Outline.} This paper is organized as follows. Section~\ref{sec:relatedwork} briefly reviews related literature. Section~\ref{sec:model} introduces the model and summarizes the main results. Sections~\ref{sec:irrevocable} and~\ref{sec:non-irr} present the results for irrevocable and non-irrevocable policy designs, respectively. Finally, Section~\ref{sec:future} concludes the paper and discusses future research directions.}

\begin{table*}[t]
\centering
\begin{tabular}{cccc}
\toprule
$\alpha$ (UB/LB) & $V_1$                   & $V_2$ & $V_{\geq 3}$ \\ \midrule
Irr                  & 1/1[Thm~\ref{thm:irrv1}]                    & 1/1[Thm~\ref{thm:irrv2}]   & $\frac{3\mathsf{min}}{\mathsf{max}}$[Thm~\ref{thm:irrv3}]/$\min\left\{\frac{1}{2}, \frac{3\mathsf{min}}{\mathsf{max}}\right\}$[Thm~\ref{thm:amchlb3}]  \\
\\[-.8em]
Irr + Na             & 1/1                    & -/$\frac{\mathsf{min}+\sqrt{\mathsf{min}(\mathsf{min}+4\mathsf{max})}}{2\mathsf{max}}$[Thm~\ref{thm:amchlb2}]   &  -/$\min\left\{\frac{1}{2}, \frac{3\mathsf{min}}{\mathsf{max}}\right\}$[Thm~\ref{thm:amchlb3}]  \\
\\[-.8em]
Na + IR              & 1/1                    & -   & $\frac{3\mathsf{min}}{\mathsf{max}}$[Thm~\ref{thm:nirrirub}]/-  \\
\\[-.8em]
Na (Greedy)                   & 1/1                    & -   & $\frac{3}{n}$[Thm~\ref{thm:nirrub}]/$\frac{2}{n}$[Thm~\ref{thm:nirrlb}] (*)   \\
\\[-.8em]
Na (Pessimistic)     & 1/1                    & -   & $\frac{3}{n}$[Thm~\ref{thm:nirrub}]/$\frac{2}{n}$[Thm~\ref{thm:nirrlb}] (*)  \\ \bottomrule
\end{tabular}
\caption{Results of Upper Bounds (UB) and Lower Bounds (LB) for worst-case Competitive Ratio in different classes of policies. Different columns indicate three subclasses of bounded monotone value functions according to how the maximum value is bounded. The last row introduces a variant of greedy players, which is called pessimistic. (*) means the bounds have an additional assumption that the number of players $n\leq (\mathsf{max/\mathsf{min}})-1$.}\label{tab:results}
\end{table*}
\section{Related Work}\label{sec:relatedwork}
\subsection{Cooperative Games with Coalition Structures}
The idea of a coalition structure is almost as old as traditional cooperative game theory itself~\citep{aumann1974cooperative}, and it remains one of the central considerations in the literature on cooperative games under various settings. For example, \citet{chalkiadakis2010cooperative} considers overlapping coalition structures, where a player can belong to multiple coalitions. \citet{kamijo2009two} characterizes a two-step Shapley value that accounts for the relationships between different coalitions. Many studies also investigate what a good coalition structure—one that yields high social welfare—might look like in specific types of cooperative games~\citep{bachrach2013optimal,meng2012cooperative,ueda2018coalition}. The main goal of these works is to determine which coalition structures should form and how value sharing can satisfy axioms related to fairness and stability.

Departing from the traditional research line, we consider an online scenario and focus more on players' actions—specifically, the coalition structure is formed through the sequential decision-making of individual players. Consequently, we study how the value-sharing policy influences their behavior.

\subsection{Online Cooperative Games}
The most similar concept to online cooperative games, which considers players arriving at different times, was first formally defined by~\citet{ge2024incentives}. The key motivation is that participants typically cannot join simultaneously, and it is often impractical to wait for all of them before sharing the value. Moreover, the set of future participants is usually unknown. As an initial work in this research direction, they focus on a setting where a player can choose to wait after realizing the game, and aim to incentivize early participation while ensuring fairness through the constraint that the value is shared in expectation equal to the Shapley value. \citet{aziz2025participation} continues working on the same setting, analyzing alternative fairness axioms. 

Our work shares a similar motivation, but we focus on a scenario where a player can either form a new sub-coalition or join an existing one, with the goal of guiding them toward forming a coalition structure that ensures high social welfare.

\subsection{Online Coalition Formations}
There is a rich body of literature on coalition formation games, and more specifically, on hedonic games~\citep{aziz2016hedonic,aziz2019fractional,bilo2022hedonic,hajdukova2006coalition,kerkmann2022altruistic,ray2015coalition}. Unlike cooperative games that focus on the value generated by different coalitions, hedonic games center on players’ preferences over various coalition compositions. When considering the online setting, their approach becomes closely related to ours when cardinal preferences are involved. For instance, \citet{flammini2021online} models player preferences using a weighted graph and studies algorithms that must irrevocably decide whether to place a newcomer into an existing coalition or create a new one, aiming to maximize the overall coalition value defined by the sum of players' cardinal preferences. Building on this, \citet{bullinger2024stability} extends the model to incorporate stability considerations, while \citet{boehmer2023causes} explores scenarios where player preferences may change dynamically over time.

The main difference between these works and ours is that we consider coalition values defined by a cooperative game rather than by player preferences.  This necessitates a \emph{value-sharing policy}, and then players will independently decide which coalition to join, rather than being directly assigned by an algorithm.

\section{Model}\label{sec:model}
Consider a cooperative game with $n$ players represented by $N =\{a_1,a_2, \ldots , a_n\}$, and a characteristic function $v: 2^N\mapsto \mathbb{R}^+$. The characteristic function specifies the value that can be created by each coalition. We assume that $v$ is monotone and bounded.

\begin{assumption}\label{asum:boundedv}
For any coalitions $S\subseteq T\subseteq N$, there are bounded values $\mathsf{min}$ and $\mathsf{max}$, such that $\mathsf{min}\leq v(S)\leq v(T)\leq \mathsf{max}$.    
\end{assumption}
Note that the bounded values are known to the policy designer. In practice, this might represent the situation where each player can at least perform a single independent task, and for each sub-project or market, there is a limit to its capacity. Let the space of all monotone and bounded characteristic functions be $V$. Assumption~\ref{asum:boundedv} holds throughout the whole paper. Players may form a coalition structure $C$, which is a partition of all players. Specifically, if $C=\{N\}$, we call it a grand coalition. The total value created by all players under the coalition structure $C$, i.e., the social welfare, is given by $ \mathsf{SW}(C\mid v) = \sum_{S\in C}v(S)$. 

\begin{figure}[t]
    \centering
    \includegraphics[width=0.75\linewidth]{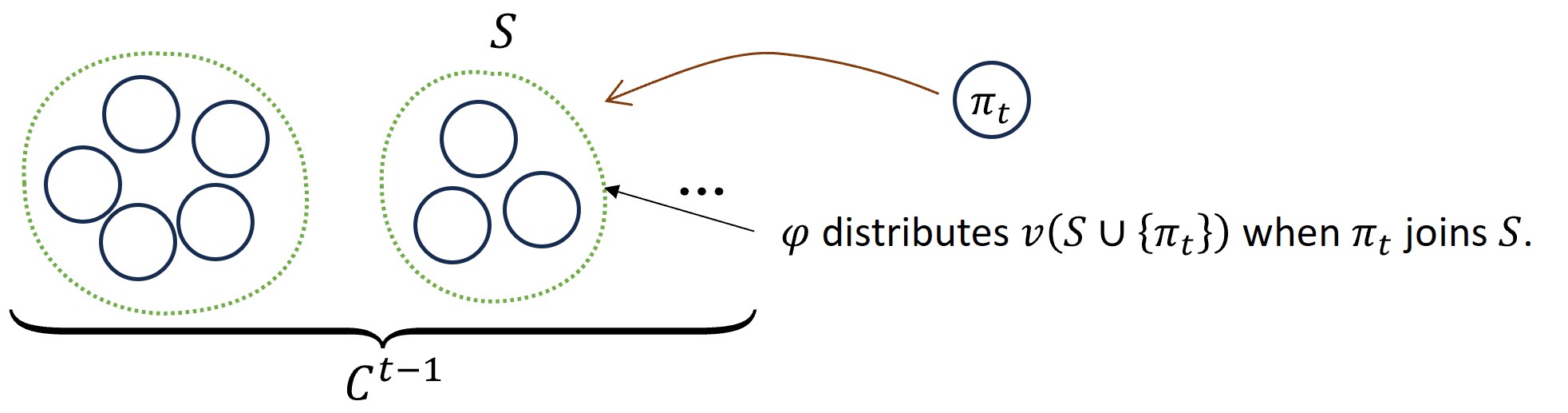}
    \caption{At each time step $t$, a player $\pi_t$ arrives and joins a coalition $S$ (could be empty, which means she creates a new one). Then we call the policy to determine the value distribution at this time.}
    \label{fig:model}
\end{figure}

In this work, we investigate a scenario where players arrive online at different times $t\in \{1,2,\dots,n\}$, and independently decide which coalition structure to form. More precisely, we assume that one player arrives in the system at each (discrete) time step. When a player arrives, she must decide whether to join an existing coalition or form a new coalition on their own. We assume that players can immediately observe the value of existing coalitions upon arrival. This scenario is practical in situations such as hiring agents in a company, where valuations are known upfront by the agents. Let $C^{t-1}$ denote the set of coalitions formed by previous players before time $t$, with $C^0=\emptyset$. Here, how $C^t$ is formed will be determined by the arrival order, how the value will be distributed, and which strategies are chosen by players (Section~\ref{sec:gp&cr} gives the details), but the following assumption holds throughout the paper.

\begin{assumption}\label{asum:one-time-decision}
   A player's decision is based solely on the information available at time $t$ and cannot be changed afterward.
\end{assumption}
Let $\pi=(\pi_1,\pi_2,\dots,\pi_n)$ be a permutation of all players, representing their arrival order. That is, the $t$-th player $\pi_t$ in $\pi$ is the one who arrives at time $t$. Denote $\pi_{\prec t}$ as the set of players who arrived before time $t$, and let $\pi^{-1}(a_i)$ represent the time at which player $a_i$ arrives.  

\begin{definition}[Sub-order and Prefix]
Given an arrival order $\pi$, and a subset of players $S \subseteq N$, denote $\pi_{|S}$ as a sub-order of players in $S$ such that they appear in the same relative order in $\pi$. Furthermore, it is said to be a prefix of $\pi$ if the players in $S$ are exactly those who arrive first under $\pi$.
\end{definition}

A player's decision depends on how the value is distributed within each coalition. A player's utility is the share of the total value created by the coalition to which they belong. The value distribution is determined by a \emph{distribution policy} $\varphi$, which is announced by the principal/mechanism designer.

\begin{definition}[Distribution policy]
    A distribution policy $\varphi$ is a function that takes as input a subset $S\subseteq N$ and its corresponding coming sub-order $\pi_{|S}$, and outputs a vector of $\mathbb{R}^{|S|}$. When $S$ forms a coalition, $\varphi$ then determines the value distribution, and for each player $i\in S$, $\varphi_i(S, \pi_{|S}\mid v)$ represents the value distributed to her\footnote{We may simply use $\varphi_i(S, \pi_{|S})$, $\varphi_i(S)$ or even $\varphi_i$ if there is no ambiguity.}.
\end{definition}

Notice that our model focuses on an online scenario, meaning that the value distribution is determined each time a new player arrives (see Figure~\ref{fig:model}). Additionally, players make decisions based on the distribution specified by the policy at the time of their arrival. Therefore, the main goal of this work is to design a suitable distribution policy to guide players' actions, ultimately maximizing social welfare as much as possible.

\subsection{Basic Attributes of a Distribution Policy}
Before modeling players' actions, we first introduce some fundamental attributes of a distribution policy. These attributes serve as basic requirements for a distribution policy or as criteria for categorizing different classes of policies. First and foremost, as the most fundamental requirement, the policy must allocate the entire value created, which is a common assumption in cooperative game scenarios~\citep{peleg2007introduction}. Therefore, throughout this paper, we only consider policies that adhere to this property, even if not explicitly stated.

\begin{definition}[Non-wastefulness]
    A distribution policy $\varphi$ is \emph{non-wasteful}\footnote{In most literature of cooperative games, it is also called efficiency~\citep{chalkiadakis2011computational}. Since in this work we seek higher social welfare, we call it non-wastefulness to avoid misunderstanding.} if it satisfies for any coalition structure $C$ of $N$, and any $S\in C$, $\sum_{i\in S} \varphi_i(S) = v(S)$.
\end{definition}

Next, we consider two additional attributes. The first concerns how the policy operates in an online manner. More specifically, when a player arrives, the policy should determine the value distribution among the existing coalitions immediately. If these decisions are irrevocable—meaning that once the value is distributed, it will not be adjusted when new players arrive—then the corresponding distribution policy is called \emph{irrevocable}. Irrevocable distribution policies not only offer a straightforward way to be implemented in an online setting (since each time, we only need to allocate the marginal contribution) but also ensure that players do not need to worry about future changes, as their utility will never decrease. In Section~\ref{sec:non-irr}, we also discuss how a non-irrevocable policy can be implemented and how players' actions might be influenced under such policies.

\begin{proposition}\label{prop:irr-def}
    If a distribution policy $\varphi$ is \emph{irrevocable} (Irr), then for any arrival order $\pi$, any $S\subseteq T\subseteq N$ such that $\pi_{|S}$ is a prefix of $\pi_{|T}$, and any $i\in S$, we have $\varphi_i(S, \pi_{|S}) \leq \varphi_i(T, \pi_{|T})$.
\end{proposition}
\begin{proof}
    By the definition of irrevocability, the value distributed to any player $i$ will not decrease when more players join. Hence, the statement holds directly.
\end{proof}
The second attribute concerns the information that the policy can utilize during the online process. In practice, the principal usually cannot predict who will join the project or what kind of contributions they will make during recruitment. This implies that a distribution policy in such scenarios must be independent of future information about $v$. We refer to such a policy as \emph{non-anticipative}.

\begin{definition}[Non-anticipation]
    A distribution policy $\varphi$ is \emph{Non-anticipative} (Na) if for any arrival order $\pi$, any $S\subseteq N$ such that $\pi_{|S}$ is a prefix of $\pi$, any $i\in S$, and any two characteristic functions $v$, $v'$, such that $v(K) = v'(K)$ for any $K\subseteq S$, we have $\varphi_i(S, \pi_{|S} \mid v) = \varphi_i(S, \pi_{|S} \mid v')$.
\end{definition}

\subsection{Greedy Players and Competitive Ratio}\label{sec:gp&cr}
We now formally model a player's decision-making process. Naturally, a player always seeks higher utility. However, she may not have information about future events. Given this information structure, we assume that a player can only make decisions based on the information available at the time of her arrival. She aims to maximize her immediate utility when joining a coalition, which implies that her actions follow a greedy strategy.

\begin{assumption}\label{asum:greedy}
    The players are greedy and aim to maximize their individual value distribution.
\end{assumption}
Formally, a player $\pi_t$ selects a coalition to join (or forms a new coalition on her own) based on the value she can obtain as if she were the last player. That is, she chooses a coalition $S^*$ that satisfies
\[ S^*\in \argmax_{S\in C^{t-1}\cup\{\emptyset\}} \varphi_{\pi_t}(S\cup\{\pi_t\}, \pi_{|S\cup\{\pi_t\}}), \]
subject to predetermined tie-breaking rules\footnote{A predetermined tie-breaking rule does not affect our results since we have already used a small perturbation to enforce the worst-case choice in our analysis, which also corresponds to the worst-case outcome in tie-breaking scenarios.}. If a player selects a coalition in this manner, she is referred to as a \emph{greedy player}\footnote{An exception may occur when considering non-irrevocable policies; see Section~\ref{sec:non-irr} for details.}. Now we can describe the outcomes under different distribution policies. Specifically, given a distribution policy $\varphi$, $C^t = C_{\mathrm{g}}^t(v,\pi \mid \varphi)$ is the coalition structure formed by greedy players, and simply write $C_{\mathrm{g}}^n$ as $C_{\mathrm{g}}$. Our primary objective, maximizing the {worst-case competitive ratio (WCR)} of social welfare can thus be expressed as follows.

\begin{definition}[Worst-case Competitive Ratio]
Let $C^*$ be a coalition structure s.t. $C^*(v)\in \arg\max_{C} \mathsf{SW}(C\mid v)$.  Then the worst-case competitive ratio  of social welfare with given policy $\varphi$ is 
\[ \alpha = \inf_{v,\pi} \frac{\mathsf{SW}(C_{\mathrm{g}}(v,\pi\mid\varphi))}{\mathsf{SW}(C^*(v))}. \]
\end{definition}

{\noindent \textbf{Note.} Throughout the rest of the paper, unless otherwise stated, by competitive ratio we mean the worst-case competitive ratio.}



\subsection{Other Properties}
Alongside the competitive ratio of social welfare, we also consider other desirable properties of the outcomes achieved by suitable distribution policies. One of the most commonly studied properties in the literature is stability~\citep{peleg2007introduction}. In our setting, a player can only determine her actions based on the information available at the time of her arrival. Therefore, when considering stability, we focus only on whether a player would be attracted to other coalitions that have already been formed when she arrives. Intuitively, stability implies that a player will not regret her decision. For convenience, for any time $t$, let $\overline{C}_g^{t} \subseteq C_g$ denote the final coalitions formed from those in $C^{t}_g$, as shown in Figure~\ref{fig:bar-notation}. We then provide the following definitions. First, it is to prevent players from regretting not creating a coalition alone.

\begin{definition}[Individual Rationality (IR)]
    A distribution policy $\varphi$ is \emph{individually rational} if for any $v$, it satisfies for any $S\in C^t_g$ with $t\leq n$, and any $i\in S$, we have $\varphi_i(S) \geq v(\{i\})$.
\end{definition}

Then, a confined version of Nash stability prevents players from regretting for not joining other coalitions when they arrive.

\begin{definition}[Temporal Nash stability (TNS)]
    A distribution policy $\varphi$ is \emph{temporal Nash stable} if for any $v$, it satisfies for any $S\in C_g$, and any $i\in S$, there does not exist another coalition $S'\in \overline{C}_g^{\pi^{-1}(i)-1}\cup \{\emptyset\}$, such that $\varphi_i(S'\cup\{i\},\pi_{|S'\cup\{i\}}) > \varphi_i(S,\pi_{|S})$.
\end{definition}

\begin{figure}[t]
    \centering
    \includegraphics[width=0.8\linewidth]{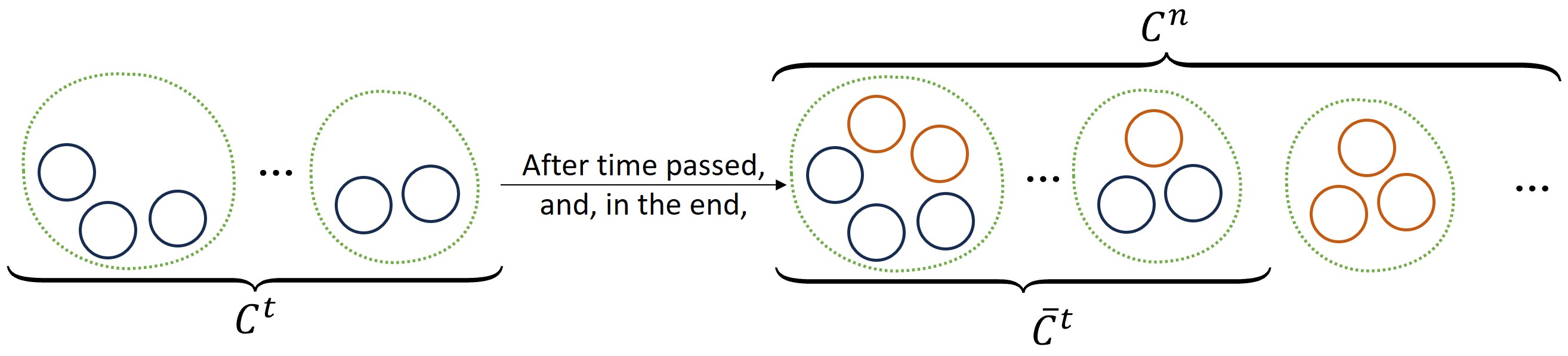}
    \caption{$\overline{C}^{t}$ contains the coalitions in the final coalition structure which have already had at least their first player at time $t$.}
    \label{fig:bar-notation}
\end{figure}

\subsection{Summary of Our Results}\label{sec:vclass}
We first summarize all the assumptions of the model:
\begin{enumerate}
    \item Assumption~\ref{asum:boundedv} and Assumption~\ref{asum:one-time-decision} are hold throughout the paper.
    \item Assumption~\ref{asum:greedy} is hold throughout all results except for some in Section~\ref{sec:non-irr}, where players can have alternative actions, referred to as \emph{pessimistic greedy players}.
\end{enumerate}
Then, Table~\ref{tab:results} summarizes the results for different classes of policies. Here, based on Assumption~\ref{asum:boundedv} and more specifically, we consider a subclass $V_{\delta}$ of $V$, where the characteristic functions further satisfy $\delta\cdot \mathsf{min}\leq \mathsf{max}< (\delta + 1)\cdot\mathsf{min}$ for some integer $\delta\geq 1$. It is clear that $V=\bigcup_{\delta=1}^{\infty} V_{\delta}$. Actually, the parameter $\delta$ captures how the marginal contribution of a new player diminishes as the coalition size increases. For example, when $\delta = 2$, it means that any coalition of size three or more has no greater total value than the players forming singleton coalitions separately. We will see that the bounds differ depending on whether {$\delta = 1$, $\delta = 2$, or $\delta \geq 3$.} More concretely, for the class of $V_1$, an optimal coalition structure can always be achieved; for the class of $V_2$, an optimal coalition structure can be achieved when \emph{non-anticipation} is not required; finally, as the most important results of this paper, for the class of {$V_{\delta}$ with $\delta \geq 3$} (we will use the notation $V_{\geq3}$ for simplicity), an upper bound of $\frac{3\mathsf{min}}{\mathsf{max}}$ holds for any class of polices, and we give a policy that achieves a competitive ratio of $\min\left\{\frac{1}{2}, \frac{3\mathsf{min}}{\mathsf{max}}\right\}$, which is near-optimal. 

At the end of this section, a useful lemma for the class $V_{\delta}$ indicates that if a coalition structure maximizes social welfare, then the number of players within each coalition is restricted. 

\begin{lemma}\label{lem:ssize}
    For a characteristic function $v\in V_{\delta}$, if a coalition structure $C^*$ maximizes the social welfare, then for any $S\in C^*$, $|S|\leq \delta$.
\end{lemma}
\begin{proof}
    We can prove this statement by contradiction. Suppose there exists a coalition $S\in C^*$ such that $|S|> \delta$. Then, if breaking $S$ into singletons, we will have
    \[ \sum_{i\in S} v(\{i\}) \geq \sum_{i\in S} \mathsf{min} = |S|\cdot \mathsf{min} \geq (\delta + 1)\cdot \mathsf{min} > \mathsf{max} \geq v(S), \]
    which contradicts the assumption that $C^*$ maximizes the social welfare. Therefore, the statement is true.
\end{proof}

\section{Irrevocable Policies}\label{sec:irrevocable}
In this section, we consider irrevocable distribution policies, which can be seen as processes that allocate marginal contributions among the existing players each time a new player joins. When a player forms a new coalition, she will instantly receive the value of being a singleton. First, we observe that such a policy is naturally individually rational, since only in the absence of better options does a greedy player form a new coalition, and irrevocability ensures non-decreasing utility. 
\begin{proposition}\label{prop:ir}
    Any irrevocable policy is individually rational.
\end{proposition}

\begin{proof}
    Consider the coalition structure $C_g$ formed by greedy players given any characteristic function $v$ and  arrival order $\pi$ under an irrevocable policy $\varphi$. We complete the proof by considering following cases.
    
    \noindent  Case 1: Suppose for any player $i\in N$, if $\{i\}\in C_g$, then $\varphi_i \geq v(\{i\})$ is satisfied.
    
    \noindent  Case 2: Suppose that   if $\{i\}\notin C_g$.   Also, suppose $i$ is not the first one to come into the coalition she belongs to. Otherwise, the situation is similar to Case 1. Since she is greedy, she will only join the coalition if she can receive at least $v(\{i\})$. Otherwise, she will form a new coalition by herself. This ensures that  $\varphi_i\geq v(\{i\})$, due to the irrevocability condition. Therefore, $\varphi$ must be IR.
   
\end{proof}
Hence, for an irrevocable policy, we only need to check whether it is TNS. In the following parts of this section, we consider the upper bounds of the competitive ratio that an irrevocable policy can achieve. We then introduce a natural class of distribution policies and evaluate their performance.

\subsection{Upper Bounds}
The upper bounds of the competitive ratio for irrevocable policies may differ depending on whether non-anticipation is required and which class of characteristic functions is considered. First, we examine the special case where characteristic functions belong to the class $V_1$. In this case, both the lower bound and the upper bound is $1$.

\begin{theorem}\label{thm:irrv1}
    For any irrevocable distribution policies, when the characteristic function $v\in V_1$, the competitive ratio must be $\alpha = 1$.
\end{theorem}
\begin{proof}
    By Lemma~\ref{lem:ssize}, we know that all players being singletons will be a coalition structure that maximizes social welfare, i.e., $C^* = \{\{i\}_{i\in N}\}$. Then, we will show that the coalition structure $C_g^t$ formed by greedy players at any time $t$ for any order $\pi$ is singletons by mathematical induction.
    \begin{itemize}
        \item \textbf{Base Case.} For time $t = 1$, the first player has no choice but to form a new singleton coalition for herself.
        \item \textbf{Induction Step.} Suppose the statement is true for any $t\leq k$. Then, consider the case of $t=k+1$. Let $\pi_{t} = j$. For any singleton $\{i\}$ formed by the first $t-1$ players, if player $j$ joins the coalition, by irrevocability, she can get at most $\mathsf{max} - v(\{i\})\leq \mathsf{max} - \mathsf{min} < \mathsf{min}$ at the time she joins. However, if she forms a singleton, she can get at least $v(\{j\})\geq \mathsf{min}$. Hence, player $j$ will also form a new coalition.
    \end{itemize}
    Then, we conclude that the coalition structure formed by greedy players, in this case, is also a collection of singletons. Therefore, the competitive ratio is $\alpha = 1$.
\end{proof}

Next, we consider an upper bound for the class $V_{\geq 3}$.

\begin{theorem}\label{thm:irrv3}
    For any irrevocable distribution policy, when the characteristic function $v\in V_{\geq 3}$, its competitive ratio cannot exceed $\frac{3\mathsf{min}}{\mathsf{max}}$.
\end{theorem}
\begin{proof}
    Since we are analyzing the worst-case competitive ratio, we can provide a single counterexample to establish an upper bound.
    Consider an instance with three players, $N = \{x, y, z\}$, an arrival order $\pi = (x, y, z)$, and the characteristic function $v \in V_{\geq 3}$:
    \begin{align*}
    & v(\{x\}) = \mathsf{min}, \quad v(\{y\}) = \mathsf{min}, \quad v(\{z\}) = \mathsf{min}, \quad v(\{x,y,z\}) = \mathsf{max}, \\
    & v(\{x,y\}) = \mathsf{min}, \qquad v(\{x,z\}) = 2\mathsf{min}, \qquad v(\{y,z\}) = 2\mathsf{min}. 
    \end{align*}

    The first player, $x$, arrives and forms a singleton. When $y$ arrives, since we focus on irrevocable policies, joining $x$ yields no gain, since $v(\{x,y\}) = \mathsf{min} = v(\{x\})$. Therefore, $y$ must also form a singleton. When $z$ arrives, joining either $x$ or $y$ results in a coalition worth $2\mathsf{min}$, which is still no better than forming a singleton, so that all options result in a total welfare of $3\mathsf{min}$. However, the optimal social welfare, which is achieved by forming the grand coalition $N$, is $\mathsf{max} \geq 3\mathsf{min}$. Thus, no irrevocable distribution policy can guarantee a competitive ratio better than $\frac{3\mathsf{min}}{\mathsf{max}}$. 
\end{proof}

Note that non-anticipation is not required above. Later in this section, we will present an irrevocable and non-anticipative policy with a competitive ratio close to this bound. This suggests that being aware of the contributions of players in the future does not help improve the competitive ratio for irrevocable policies in the class $V_{\geq 3}$. Intuitively, it is because for greedy players whose actions are solely determined by immediate rewards, knowing about future players sometimes cannot affect their decisions.
However, when the characteristic functions belong to class $V_2$, being aware of future information may help: 
if we relax the requirement of non-anticipation, then there exists a distribution policy with a competitive ratio of $1$. 

\begin{theorem}\label{thm:irrv2}
    For characteristic functions $v\in V_{2}$, there exists an irrevocable distribution policy that has a competitive ratio $\alpha = 1$.
\end{theorem}
\begin{proof}
    By Lemma~\ref{lem:ssize}, any coalition structure $C^*$ that maximizes the social welfare for a characteristic function $v\in V_{2}$ only contains singletons and pairs. Now, consider the following distribution policy:
    \begin{itemize}
        \item Choose a coalition structure $C^*$ that maximizes social welfare.
        \item For any time $t$, if the player $\pi_t$ joins with an existing coalition $S\neq \emptyset$ such that $S\cup\{\pi_t\}\notin C^*$, then her marginal contribution will be distributed to the last player before her $S$, i.e., she receives nothing; if $S\cup\{\pi_t\}\in C^*$ or she creates a new coalition, then she will receive all her marginal contribution.
    \end{itemize}
    With this policy, $C_g$ will always be $C^*$, and hence $\alpha = 1$. (i) Consider a singleton $\{i\}\in C^*$, she will never join with others since she will get nothing, and other players will also not join with her as they will get nothing. (ii) Consider a pair of players $\{i,j\}\in C^*$, and w.l.o.g, assume $\pi^{-1}(i) < \pi^{-1}(j)$. Any other player will not join with $i$ as she will get nothing. Provided player $i$ is still a singleton when player $j$ comes, then joining with player $i$ is the best choice: joining with other player will get nothing, and being alone cannot be better because $v(\{i,j\})\geq v(\{i\}) + v(\{j\})$ leads to $v(\{i,j\}) - v(\{i\})\geq v(\{j\})$. 
\end{proof}

The key difference that occurs in the class of $V_2$ is that players have no ``hidden contributions'', which are only revealed when the coalition is sufficiently large, allowing us to force them to form the ``correct'' pairs. However, this can only be done by knowing what the ``correct'' pairs look like. Therefore, with the requirement of non-anticipation, we can no longer achieve this.


\subsection{Allocating Marginal Contributions}
Now we are going to show specific distribution policies. First of all, one of the most direct ways to construct an irrevocable policy is simply allocating the marginal contribution to each newly joined player.

\begin{definition}[Allocating Marginal Contributions (AMC)]
    A distribution policy $\varphi$ is \emph{allocating marginal contributions} if it satisfies for each $i\in S$,
    $ \varphi_i(S, \pi_{|S}) = v(\pi_{\prec\pi^{-1}(i)}\cap S\cup\{i\}) - v(\pi_{\prec\pi^{-1}(i)}\cap S)$.
\end{definition}

Clearly, AMC policy is also non-anticipative. It is easy to observe that it also satisfies TNS, as in the following proposition. 

\begin{proposition}\label{prop:amcNs}
    AMC policy is temporal Nash stable.
\end{proposition}
\begin{proof}
    Since the utility of any player will never change after she has already joined, 
    then the policy is naturally TNS by definition.
\end{proof}

On the other hand, we also observe that a coalition formed by greedy players under this policy has a bounded size similar to an optimal one. As a result, this simple policy has a competitive ratio of $\alpha = \frac{2\mathsf{min}}{\mathsf{max}}$ in the class $V_{\geq 2}$.

\begin{lemma}\label{lem:amcsize}
    With AMC policy, the coalition structure $C_g$ formed by greedy players under any characteristic function $v\in V_{\delta}$ and any arrival order $\pi$ must satisfy that for any $S\in C_g$, $|S|\leq \delta$.
\end{lemma}
\begin{proof}
    We first observe that for $S\in C^t$ at any time $t$, we must have $v(S)\geq |S|\cdot\mathsf{min}$. This is because, for each player in $S$, the marginal contribution she brought is no less than $\mathsf{min}$; otherwise, she can choose to form a new coalition. Next, we show the statement by contradiction. Suppose that there exists a coalition $S\in C_g$ such that $|S|>\delta$. Then, $v(S)\geq |S|\cdot\mathsf{min}\geq (\delta+1) \cdot \mathsf{min}$, which contradicts to the assumption that $v\in V_{\delta}$. Therefore, there cannot be a coalition $S\in C_g$ such that $|S|>\delta$, so that the statement is true.
\end{proof}

\begin{theorem}\label{thm:amclb}
    AMC policy has a competitive ratio of $\alpha = \frac{2\mathsf{min}}{\mathsf{max}}$ when the characteristic function belonging to $V_{\geq 2}$.
\end{theorem}

Theorem~\ref{thm:amclb} is a direct corollary from Theorem~\ref{thm:amchlb2}  and~\ref{thm:amchlb3}, so we omit the proof here. Following, we give an example where the ratio is exactly $\frac{2\mathsf{min}}{\mathsf{max}}$. 

\begin{example}\label{ex:amcbound}
    Consider the following characteristic function $v\in V_\delta$ with $n=2m\delta$ players such that $v(\{a_i\})=\mathsf{min}$ for any $i\in N$, $v(\{a_i,a_{i+1},\dots,a_{i+k}\}) - v(\{a_i,a_{i+1},\dots,a_{i+k-1}\})=\mathsf{min} + \epsilon$ with $\epsilon\to 0$ for any $k<\delta$, and $v(\{a_i,a_{i+\delta}\}) = \mathsf{max}$. Let $\pi = (a_1,a_2,\dots,a_n)$. Then, for $1< t\leq \delta$, the player $\pi_t$ will always join the existing coalition; for $t=\delta + 1$, the player $\pi_{t+1}$ has a marginal contribution to the existing coalition be less than $\mathsf{max}-\delta\cdot\mathsf{min}<\mathsf{min}$, so that she will form a new coalition. The following players will do similar things as the previous players and finally form $m$ coalitions with the value of about $\delta\cdot\mathsf{min}$. However, the best coalition structure that maximizes social welfare is to have pairs that can create a value of $\mathsf{max}$. Therefore, in this instance, AMC policy has a competitive ratio of $\alpha = \frac{n\cdot\mathsf{min}}{\frac{n}{2}\cdot\mathsf{max}} = \frac{2\mathsf{min}}{\mathsf{max}}$.
\end{example}

\subsection{AMC beyond Threshold}
We observe that in the worst-case scenario for the AMC policy described above, players with smaller marginal contributions who arrive earlier may occupy the positions of better partners. Inspired by this observation, we now consider setting a threshold to prevent players with insufficient marginal contributions from joining an existing coalition. To achieve this for greedy players using an irrevocable policy, we propose the following class of distribution methods.

\begin{definition}[Allocating Marginal Contributions beyond Threshold $h$ (AMC-$h$)]
    A distribution policy $\varphi$ is \emph{allocating marginal contributions beyond threshold $h$} if it satisfies for each $i\in S$, denoting the marginal contribution when she joins as
    \[ \mathsf{MC}_i(S, \pi_{|S}) = v(\pi_{\prec\pi^{-1}(i)}\cap S\cup\{i\}) - v(\pi_{\prec\pi^{-1}(i)}\cap S), \]
    then $\varphi$ updates the value distribution as follows:
    \begin{itemize}
        \item If $\mathsf{MC}_i(S, \pi_{|S})\leq h$, then all these marginal contributions will be allocated to the last player before $i$ in $S$;
        \item If $\mathsf{MC}_i(S, \pi_{|S}) > h$, then the last player before $i$ in $S$ will additionally receive $h$ from the marginal contribution, and $i$ herself will get the remaining.
        \item Specially, when $\pi_{\prec\pi^{-1}(i)}\cap S = \emptyset$, i.e., $i$ is the first player who forms $S$, then the last player before $i$ is treated as herself.
    \end{itemize}
\end{definition}

Clearly, AMC-$h$ policy is also irrevocable and non-anticipative by its definition. With this policy, greedy players will first find an existing coalition $S$ that maximizes $\max\{v(S\cup\{i\}) - v(S) - h, 0\}$, and compare the value with $v(\{i\})$: if the former is larger, she will join that coalition; otherwise, she will form a new coalition. The final value distribution in a coalition follows the AMC policy, except that the first player additionally receives $h$ from the last player as a bonus for initiating a new coalition that others would like to join.

With the introduction of the threshold, the policy is no longer TNS necessarily. More precisely, it maintains the property of TNS only when it coincides with AMC policy.

\begin{proposition}\label{prop:h-tns}
    AMC-$h$ policy is TNS if and only if $h = 0$.
\end{proposition}
\begin{proof}
    For the sufficient side, the statement is true due to Proposition~\ref{prop:amcNs}. We only need to show that $h=0$ is also a necessary condition. If an AMC-$h$ policy is TNS, then consider a player who is the last one coming in the coalition she belongs to, but not the last player in $\pi$. That means, if she did not join the current coalition, there exists a case where she will not be the last one in another coalition. Let the player's utility be $\mathsf{MC} - h$, and the utility in the second case is $\mathsf{MC}' - h + h = \mathsf{MC}'$, with $\mathsf{MC}'\leq \mathsf{MC}$. Since the policy is TNS, then for any $\mathsf{MC}'\leq \mathsf{MC}$, $\mathsf{MC} - h\geq \mathsf{MC}'$ should be satisfied. Therefore, it implies that $h$ can only be $0$.
\end{proof}

For the competitive ratio of the AMC-$h$ policy, we will see that it differs between the class $V_2$ and the class $V_{\geq 3}$. Before illustrating these differences, we first give a generalization of Lemma~\ref{lem:amcsize}.

\begin{lemma}\label{lem:amchsize}
    With AMC-$h$ policy, the coalition structure $C_g$ formed by greedy players under any characteristic function $v\in V_\delta$ and any arrival order $\pi$ must satisfies that for any $S\in C_g$, $|S|\leq\left\lceil \frac{\delta\cdot\mathsf{min}}{\mathsf{min} + h} \right\rceil$.
\end{lemma}
\begin{proof}
    To show the statement, we first observe that for any $S\in C^t$ at any time $t$, we must have $v(S)\geq \mathsf{min} + (|S|-1)(\mathsf{min} + h)$. This is because for each player in $S$ except for the first one, the marginal contribution she brought is no less than $\mathsf{min}+h$; otherwise, she can choose to form a new coalition. Then, we show the statement by contradiction. Suppose that there exists a coalition $S\in C_g$ such that $|S|>\left\lceil \frac{\delta\cdot\mathsf{min}}{\mathsf{min} + h} \right\rceil$. Then we have
    \begin{align*}
        v(S) & \geq \mathsf{min} + \left(\left\lceil \frac{\delta\cdot\mathsf{min}}{\mathsf{min} + h} \right\rceil +1 -1\right)(\mathsf{min} + h) \\
        & \geq \mathsf{min} + \frac{\delta\cdot\mathsf{min}}{\mathsf{min} + h} \cdot (\mathsf{min} + h) = (\delta + 1)\cdot \mathsf{min},
    \end{align*}
    which contradicts the assumption that $v\in V_\delta$. Therefore, there cannot be a coalition $S\in C_g$ such that $|S|>\left\lceil \frac{\delta\cdot\mathsf{min}}{\mathsf{min} + h} \right\rceil$, so that the statement is true.
\end{proof}
Next, we determine the competitive ratio of the AMC-$h$ policy.

\begin{theorem}\label{thm:amchlb2}
    The AMC-$h$ policy has a competitive ratio of $\frac{\mathsf{min}+\sqrt{\mathsf{min}(\mathsf{min}+4\mathsf{max})}}{2\mathsf{max}}$ for maximizing social welfare when the characteristic function $v\in V_{2}$ when $h=\sqrt{\mathsf{min}(\mathsf{min}+4\mathsf{max})}-3\mathsf{min}$.
\end{theorem}
\begin{proof}
    If $h > \mathsf{max} - 2\mathsf{min}$, then the greedy players only form singletons since for any two players $a_i$ and $a_j$, $v(\{a_i,a_j\}) - v(\{a_i\}) - h \leq \mathsf{max} - \mathsf{min} - h < h + 2\mathsf{min} - \mathsf{min} - h = \mathsf{min}\leq v(\{a_j\})$. Hence, the worst case occurs when all pairs of two players can create value $\mathsf{max}$, the ratio over the optimal social welfare will be $\frac{2\mathsf{min}}{\mathsf{max}}$, so that such an $h$ will make AMC-$h$ policy has a competitive ratio of $\frac{2\mathsf{min}}{\mathsf{max}}$. 
    
    If $h \leq \mathsf{max} - 2\mathsf{min}$, given any $v\in V_{2}$ and any $\pi$, let the $C_g$ be the coalition structure formed by greedy players, and $C^*$ be a coalition structure that maximizes social welfare. Then we can split both coalition structure to the same number of parts with the same sets of players like $C_g = (C_1, C_2, \dots, C_m)$ and $C^* = (C^*_1, C^*_2, \dots, C^*_m)$, such that $\bigcup_{S\in C_k}S = \bigcup_{S\in C^*_k}S$ for any $1\leq k\leq m$, and each pair $(C_k, C^*_k)$ cannot be split to smaller parts with the same sets of players. Then, 
    $ \frac{\mathsf{SW}(C_g)}{\mathsf{SW}(C^*)} \geq \min\left\{ \frac{\mathsf{SW}(C_1)}{\mathsf{SW}(C^*_1)}, \frac{\mathsf{SW}(C_2)}{\mathsf{SW}(C^*_2)}, \dots, \frac{\mathsf{SW}(C_m)}{\mathsf{SW}(C^*_m)} \right\} $.
    
    Hence, we only need to the lower bound of $\frac{\mathsf{SW}(C_k)}{\mathsf{SW}(C^*_k)}$, and it has following cases when $C_k$ is different from $C^*_k$:
    \begin{itemize}
        \item If $\begin{cases}
            C_k = S_1, S_2, \dots, S_q \\
            C_k^* = S_1\cup S_2\cup\dots\cup S_q
        \end{cases}$ 
        
        with $|S_1| + |S_2| + \dots + |S_q| = r\geq 2$, then according to Lemma~\ref{lem:ssize}, we have $|S_1\cup\dots\cup S_q|\leq\delta =  2$, so that $r=2$ and $q = 2$. Let $S_1 = \{a_1\}$ and $S_2 = \{a_2\}$; then this case only occur when $v(\{a_1, a_2\})\leq v(\{a_1\}) + v(\{a_2\}) + h$. Hence, we have $\frac{\mathsf{SW}(C_k)}{\mathsf{SW}(C^*_k)}\geq \frac{v(\{a_1\}) + v(\{a_2\})}{v(\{a_1\}) + v(\{a_2\}) + h}\geq \frac{2\mathsf{min}}{2\mathsf{min} + h}$.
        \item If $\begin{cases}
            C_k = S_1\cup S_2\cup\dots\cup S_q \\
            C_k^* = S_1, S_2, \dots, S_q
        \end{cases}$ 
        
        with $|S_1| + |S_2| + \dots + |S_q| = r\geq 2$, then according to Lemma~\ref{lem:amchsize}, we have $|S_1\cup S_2\cup\dots\cup S_q|\leq\left\lceil \frac{\delta\cdot\mathsf{min}}{\mathsf{min} + h} \right\rceil = \left\lceil \frac{2\mathsf{min}}{\mathsf{min} + h} \right\rceil \leq 2$, so that $r = 2$ and $q = 2$. Let $S_1 = \{a_1\}$ and $S_2 = \{a_2\}$; then this case only occur when $v(\{a_1, a_2\})\geq v(\{a_1\}) + v(\{a_2\}) + h$. However, we also have $\mathsf{SW}(C_k^*)\geq \mathsf{SW}(C_k)$, which means $v(\{a_1\}) + v(\{a_2\}) \geq v(\{a_1,a_2\})$. Hence, this case cannot happen for any $h>0$, and for $h = 0$, the ratio in this case will be $1$. 
        \item If $\begin{cases}
            C_k = S_1, S_2, \dots, S_q \\
            C_k^* = S_1', S_2', \dots, S_{q'}'
        \end{cases}$,  
        
        with $q>1$, $q'>1$ and $|S_1| + |S_2| + \dots + |S_q| = r> 2$, then similarly, according to Lemma~\ref{lem:ssize} and Lemma~\ref{lem:amchsize}, each coalition in $C_k$ or $C_k^*$ should have a size no larger than $2$. If any coalition in $C_k$ or $C_k^*$ has exactly size of $2$, then we have $\frac{\mathsf{SW}(C_k)}{\mathsf{SW}(C^*_k)}\geq \frac{(r/2)\cdot(2\mathsf{min}+h)}{(r/2)\cdot\mathsf{max}} = \frac{2\mathsf{min}+h}{\mathsf{max}}$. Note that the partition is of the smallest size. If there is a singleton in $C_k$, denoted by $\{a_1\}$, then there must be a pair in $C_k^*$ contains $a_1$, denoted by $\{a_1, a_2\}$. If $a_2$ is a singleton in $C_k$, then there will be no more any other players; otherwise, there should be a player denoted by $a_3$ such that $\{a_2, a_3\}\in C_k$. Continue the process, and finally we can know that there are exactly {two singletons} among all coalitions in $C_k$ and $C_k^*$. It is the same if we start to consider a singleton in $C_k^*$.  {Hence, by the symmetry of indices, there are three cases}:
        \begin{itemize}
            \item $\begin{cases}
            C_k = \{a_1\}, \{a_2,a_3\}, \{a_4,a_5\}, \dots, \{a_{2x-2},a_{2x-1}\}, \{a_{2x}\} \\
            C_k^* = \{a_1,a_2\}, \{a_3,a_4\}, \dots, \{a_{2x-1}, a_{2x}\}
        \end{cases}$ 
        
        with $r = 2x$ and $x\geq 2$: the worst case happens when all coalitions in $C_k$ get their minimal possible value, and all pairs in $C_k^*$ get their maximal possible value. Then, we have $\frac{\mathsf{SW}(C_k)}{\mathsf{SW}(C^*_k)}\geq \frac{2\mathsf{min}+(x-1)\cdot(2\mathsf{min}+h)}{x\cdot\mathsf{max}} \geq \frac{4\mathsf{min} + h}{2\mathsf{max}}$.
        \item $\begin{cases}
            C_k = \{a_1\}, \{a_2,a_3\}, \{a_4,a_5\}, \dots, \{a_{2x},a_{2x+1}\} \\
            C_k^* = \{a_1,a_2\}, \{a_3,a_4\} \dots, \{a_{2x-1}, a_{2x}\}, \{a_{2x+1}\}
        \end{cases}$ 
        
        with $r = 2x + 1$: since player $a_{2x+1}$ can be paired in $C_k$, then it should satisfy $v(\{a_{2x+1}\})\leq v(\{a_{2x},a_{2x+1}\}) - v(\{a_{2x}\}) - h \leq v(\{a_{2x},a_{2x+1}\}) - \mathsf{min} - h$. 
        \\
    Then, we have $\frac{\mathsf{SW}(C_k)}{\mathsf{SW}(C^*_k)}\geq \frac{\mathsf{min}+(x-1)\cdot(2\mathsf{min}+h) + v(\{a_{2x},a_{2x+1}\})}{x\cdot\mathsf{max} + v(\{a_{2x},a_{2x+1}\}) - \mathsf{min} - h} \geq \frac{\mathsf{min}+x\cdot(2\mathsf{min}+h)}{x\cdot\mathsf{max} + \mathsf{min}} \geq \frac{2\mathsf{min} + h}{\mathsf{max}}$.
        \item $\begin{cases}
            C_k = \{a_1,a_2\}, \{a_3,a_4\}, \dots, \{a_{2x-1}, a_{2x}\} \\
            C_k^* = \{a_1\}, \{a_2,a_3\}, \{a_4,a_5\}, \dots, \{a_{2x-2},a_{2x-1}\}, \{a_{2x}\}
        \end{cases}$

        with $r = 2x$: similar to above, we have $\frac{\mathsf{SW}(C_k)}{\mathsf{SW}(C^*_k)} \geq \frac{x\cdot(2\mathsf{min}+h)}{(x-1)\cdot\mathsf{max} + 2\mathsf{min}} \geq \frac{x\cdot(2\mathsf{min}+h)}{x\cdot\mathsf{max}} \geq \frac{2\mathsf{min} + h}{\mathsf{max}}$.
        \end{itemize}
    \end{itemize}
    Taking all cases together, when $h\leq \mathsf{max} - 2\mathsf{min}$, the ratio $\alpha$ satisfies
    $\alpha\geq \min\left\{ \frac{2\mathsf{min}}{2\mathsf{min}+h}, \frac{4\mathsf{min}+h}{2\mathsf{max}} \right\}$.
    
    \noindent Therefore, the best $\alpha$ can be achieved is $\frac{\mathsf{min}+\sqrt{\mathsf{min}(\mathsf{min}+4\mathsf{max})}}{2\mathsf{max}}$ when $h=\sqrt{\mathsf{min}(\mathsf{min}+4\mathsf{max})}-3\mathsf{min}$ ({this choice of $h$ makes that $h\leq\mathsf{max} - 2\mathsf{min}$} always be hold).
\end{proof}

\begin{theorem}\label{thm:amchlb3}
    The AMC-$h$ policy has a competitive ratio of $\min\left\{\frac{1}{2}, \frac{3\mathsf{min}}{\mathsf{max}}\right\}$ for maximizing social welfare when the characteristic function belonging to the class $V_{\delta}$ with $\delta \geq 3$, when $h = \begin{cases}
        \mathsf{min} & \mathsf{max} < 4\mathsf{min} \text{ (i.e., }v\in V_3 \text{)}; \\
        2\mathsf{min} & \mathsf{max} \geq 4\mathsf{min} \text{ (i.e., }v\in V_{\geq 4} \text{)}.
    \end{cases}$
\end{theorem}
\begin{proof}
    If $h > \mathsf{max} - 2\mathsf{min}$, then the greedy players can only form singletons since for any two players $a_i$ and $a_j$, $v(\{a_i,a_j\}) - v(\{a_i\}) - h \leq \mathsf{max} - \mathsf{min} - h < h + 2\mathsf{min} - \mathsf{min} - h = \mathsf{min}\leq v(\{a_j\})$. Hence, the worst case occurs when all pairs of two players can create value $\mathsf{max}$, the ratio over the optimal social welfare will be $\frac{2\mathsf{min}}{\mathsf{max}}$, so that such an $h$ will make AMC-$h$ policy has a competitive ratio of $\frac{2\mathsf{min}}{\mathsf{max}}$. Then, we will focus on the opposite case.
    
    If $h \leq \mathsf{max} - 2\mathsf{min}$, given any $v\in V_{\geq 3}$ and any $\pi$, let the $C_g$ be the coalition structure formed by players, and $C^*$ be a coalition structure that maximizes social welfare. Then we can split both coalition structure to the same number of parts with the same sets of players like $C_g = (C_1, C_2, \dots, C_m)$ and $C^* = (C^*_1, C^*_2, \dots, C^*_m)$, such that $\bigcup_{S\in C_k}S = \bigcup_{S\in C^*_k}S$ for any $1\leq k\leq m$, and each pair $(C_k, C^*_k)$ cannot be split to smaller parts with the same sets of players. Then, we have
    \[ \frac{\mathsf{SW}(C_g)}{\mathsf{SW}(C^*)} \geq \min\left\{ \frac{\mathsf{SW}(C_1)}{\mathsf{SW}(C^*_1)}, \frac{\mathsf{SW}(C_2)}{\mathsf{SW}(C^*_2)}, \dots, \frac{\mathsf{SW}(C_m)}{\mathsf{SW}(C^*_m)} \right\}. \]
    Hence, we only need to the lower bound of $\frac{\mathsf{SW}(C_k)}{\mathsf{SW}(C^*_k)}$, and it has following cases when $C_k$ is different from $C^*_k$:
    \begin{itemize}
        \item If $\begin{cases}
            C_k = S_1, S_2, \dots, S_q \\
            C_k^* = S_1\cup S_2\cup\dots\cup S_q
        \end{cases}$ 
        
        with $|S_1| + |S_2| + \dots + |S_q| = r\geq 2$, then according to Lemma~\ref{lem:amchsize}, we have $v(S_1) + v(S_2) + \dots + v(S_q) \geq r \cdot \mathsf{min} + (r-q)\cdot h$. If $r = 2$, i.e., $S_1 = \{a_1\}$ and $S_2 = \{a_2\}$, this case only occur when $v(\{a_1, a_2\})\leq v(\{a_1\}) + v(\{a_2\}) + h$; then, we have $\frac{\mathsf{SW}(C_k)}{\mathsf{SW}(C^*_k)}\geq \frac{v(\{a_1\}) + v(\{a_2\})}{v(\{a_1\}) + v(\{a_2\}) + h}\geq \frac{2\mathsf{min}}{2\mathsf{min} + h}$. If $r\geq 3$, then we have $\frac{\mathsf{SW}(C_k)}{\mathsf{SW}(C^*_k)}\geq \frac{r \cdot \mathsf{min} + (r-q)\cdot h}{\mathsf{max}}\geq \frac{r\cdot\mathsf{min}}{\mathsf{max}}\geq \frac{3\mathsf{min}}{\mathsf{max}}$. Hence, in this case, $\frac{\mathsf{SW}(C_k)}{\mathsf{SW}(C^*_k)}\geq \min\left\{ \frac{2\mathsf{min}}{2\mathsf{min} + h}, \frac{3\mathsf{min}}{\mathsf{max}} \right\}$.
        \item If $\begin{cases}
            C_k = S_1\cup S_2\cup\dots\cup S_q \\
            C_k^* = S_1, S_2, \dots, S_q
        \end{cases}$ 
        
        with $|S_1| + |S_2| + \dots + |S_q| = r\geq 2$, then similarly according to Lemma~\ref{lem:amchsize}, we have $v(S_1\cup S_2\cup\dots\cup S_q)\geq r\cdot \mathsf{min} + (r-1)\cdot h$. Moreover, suppose there are $p\leq r$ singletons in $C_k^*$, and w.l.o.g, let $S_1$ to $S_p$ be singletons. According to AMC-$h$, when a singleton in $C_k^*$ comes, since she chooses to join with others rather than forming a new coalition, her marginal contribution must be no less than the value of being singleton plus $h$ (except for the one coming the first). Then, we have $\frac{\mathsf{SW}(C_k)}{\mathsf{SW}(C^*_k)}\geq \frac{(r-p)\cdot\mathsf{min} + (r-1)\cdot h + v(S_1) + \dots + v(S_p)}{\frac{r-p}{2}\mathsf{max} + v(S_1) + \dots + v(S_p)}\geq \frac{(r-p)\cdot\mathsf{min} + (r-1)\cdot h}{\frac{r-p}{2}\mathsf{max}} = \frac{2\mathsf{min}}{\mathsf{max}} + \frac{2(r-1)\cdot h}{(r-p)\cdot\mathsf{max}}$, where $\frac{2(r-1)\cdot h}{(r-p)\cdot\mathsf{max}} \geq \frac{2(r-1)\cdot h}{r\cdot\mathsf{max}} \geq 2\cdot\frac{1}{2}\cdot\frac{h}{\mathsf{max}} = \frac{h}{\mathsf{max}}$. Hence, $\frac{\mathsf{SW}(C_k)}{\mathsf{SW}(C^*_k)}\geq \frac{2\mathsf{min} + h}{\mathsf{max}}$.
        \item If $\begin{cases}
            C_k = S_1, S_2, \dots, S_q \\
            C_k^* = S_1', S_2', \dots, S_{q'}'
        \end{cases}$,

        with $q>1$, $q'>1$ and $|S_1| + |S_2| + \dots + |S_q| = r> 2$. Then, there are the following possibilities in this case:
        \begin{itemize}
            \item If there is no singleton in $C_k$, and suppose there are $p\leq r$ singletons in $C_k^*$ (w.l.o.g., let $S_1'$ to $S_p'$ be singletons), then for each singleton $S_i'$, there must be a coalition $S_j\in C_k$ that contains her, and her marginal contribution is at least $v(S_i') + h$ or $v(S_i')$ if she comes the first among $S_j$. Hence, we have $\frac{\mathsf{SW}(C_k)}{\mathsf{SW}(C^*_k)}\geq \frac{(r-p)\cdot\mathsf{min} + (r-q)\cdot h + v(S_1) + \dots + v(S_p)}{\frac{r-p}{2}\mathsf{max} + v(S_1) + \dots + v(S_p)}\geq \frac{(r-p)\cdot\mathsf{min} + \frac{r}{2}\cdot h}{\frac{r-p}{2}\mathsf{max}} \geq \frac{2\mathsf{min}}{\mathsf{max}} + \frac{r\cdot h}{(r-p)\cdot\mathsf{max}} \geq \frac{2\mathsf{min} + h}{\mathsf{max}}$.
            \item If there is exactly one singleton in $C_k$, also suppose there are $p\leq r$ singletons in $C_k^*$ (w.l.o.g., let $S_1'$ to $S_p'$ be singletons). Then, similar to the analysis in the first case, when $r$ is odd and $p\geq 1$ is odd, $\frac{\mathsf{SW}(C_k)}{\mathsf{SW}(C^*_k)}\geq \frac{(r-p)\cdot\mathsf{min} + \frac{r-1}{2}\cdot h}{\frac{r-p}{2}\mathsf{max}} \geq \frac{2\mathsf{min}}{\mathsf{max}} + \frac{(r-1)\cdot h}{(r-p)\cdot\mathsf{max}} \geq \frac{2\mathsf{min} + h}{\mathsf{max}}$; when $r$ is odd and $p\geq 0$ is even, $\frac{\mathsf{SW}(C_k)}{\mathsf{SW}(C^*_k)}\geq \frac{(r-p)\cdot\mathsf{min} + \frac{r-1}{2}\cdot h}{\frac{r-p-1}{2}\mathsf{max}} \geq \frac{2(r-p)\mathsf{min}}{(r-p-1)\mathsf{max}} + \frac{(r-1)\cdot h}{(r-p-1)\cdot\mathsf{max}} \geq \frac{2\mathsf{min} + h}{\mathsf{max}}$; when $r$ is even and $p$ is odd, $\frac{\mathsf{SW}(C_k)}{\mathsf{SW}(C^*_k)}\geq \frac{(r-p)\cdot\mathsf{min} + \frac{r}{2}\cdot h}{\frac{r-p}{2}\mathsf{max}} \geq \frac{2\mathsf{min}}{\mathsf{max}} + \frac{r\cdot h}{(r-p)\cdot\mathsf{max}} \geq \frac{2\mathsf{min} + h}{\mathsf{max}}$; when $r$ is even and $p$ is even, $\frac{\mathsf{SW}(C_k)}{\mathsf{SW}(C^*_k)}\geq \frac{(r-p)\cdot\mathsf{min} + \frac{r}{2}\cdot h}{\frac{r-p-1}{2}\mathsf{max}} \geq \frac{2(r-p)\mathsf{min}}{(r-p-1)\mathsf{max}} + \frac{r\cdot h}{(r-p)\cdot\mathsf{max}} \geq \frac{2\mathsf{min} + h}{\mathsf{max}}$.
            \item If there are more than one singletons in $C_k$, then suppose there are $1<p<q$ singletons in $C_k$ (w.l.o.g., let $S_1$ to $S_p$ be singletons, and denote the players as $a_1$ to $a_p$). Suppose players $a_1$ to $a_p$ have been divided into $p'$ coalitions in $C_k^*$ with $p'\leq p$, and w.l.o.g., let them be $S_1'$, $S_2'$, $\dots$, $S_{p'}'$. For each coalition among them, there is at least one player who is not being singleton in $C_k$, so they are all not singletons. Denote these additional players as $a_{p+1}$ to $a_{p+\bar{p}}$ with $\bar{p}\geq p'$. Then, consider the coalitions in $C_k$, denoting them as $S_{p+1}$ to $S_{p+\tilde{p}}$ with $\tilde{p}\leq \bar{p}$. Until now, we have grouped the coalitions in $C_k$ and $C_k^*$ as: $C_k = (S_1,\dots,S_p), (S_{p+1},\dots,S_{p+\tilde{p}}), (S_{p+\tilde{p}+1},\dots, S_q)$ and $C_k^* = (S_1',\dots, S_{p'}'), (S_{p'+1}',\dots, S_{q'}')$, where the first group of $C_k$ are singletons, the first group of $C_k^*$ contains all these players, the first two groups of $C_k$ contains all players in the first group of $C_k^*$, and there is no singletons except for the first group of $C_k$ and the second group of $C_k^*$.

            Now we focus on the second group of $C_k$ and try to split it. For each coalition $S\in (S_{p+1},\dots,S_{p+\tilde{p}})$, there are two kinds of players. One belongs to $a_{p+1}$ to $a_{p+\bar{p}}$ and is contained by the first group of $C_k^*$; the other is the opposite. Let the two kinds of players form the sub-coalition $S^L$ and $S^R$ respectively ($S^R$ could be empty). Then, we have $v(S)\geq \left(|S^L|\cdot\mathsf{min} + (|S^L| - 1)h\right) + h + \left(|S^R|\cdot\mathsf{min} + (|S^R| - 1)h\right)$; hence, by taking their minimum value, we can split $S$ into two coalitions $S^L$ and $S^R$ with an additional value of $h$. Here, we use a trick that when $S^R$ is non-empty: if $|S^R| = 1$, let the value of $h$ follows with $S^R$; otherwise it follows with $S^L$ to be their new minimum value. The competitive ratio now is $\alpha\geq \min\left\{ \frac{\mathsf{SW}(S_1,\dots,S_p,S^L_{p+1},\dots,S^L_{p+\tilde{p}})}{\mathsf{SW}(S_1',\dots,S_p')}, \frac{\mathsf{SW}(S^R_{p+1},\dots,S^R_{p+\tilde{p}},S_{p+\tilde{p}+1},\dots,S_q)}{\mathsf{SW}(S_{p'+1}',\dots,S_{q'}')}\right\}$ with $S^L$ and $S^R$ take the above modified minimum value. Then, we can check those two parts.
            \begin{itemize}
                \item For the first part, we can notice that we can do the same way of splitting with the smallest part with the same set of players as that for original $C_g$ and $C^*$. Notice that in these parts of $C^*$, there is no singleton. Then, there are following cases for any one of the smallest part, denoted by $C_s$ and $C_s^*$ with $r_s$ players. (i) If all the coalitions in $C_s^*$ have size of no less than $3$, then we must have $\frac{\mathsf{SW}(C_s)}{\mathsf{SW}(C^*_s)} \geq \frac{r_s\cdot\mathsf{min}}{(r_s/3)\cdot \mathsf{max}} = \frac{3\mathsf{min}}{\mathsf{max}}$. (ii) If $C_s^*$ is mixed of coalitions with size no less than $3$ and size $2$, consider a coalition in $C_s^*$ has size $y+1$, $y\geq 2$. Suppose there are $x$ singletons belonging to $(S_1,\dots,S_p)$ in $C_s$. Then, if $x> y$, there are at most $y$ other pairs in $C_s^*$ whose members are paired in $C_s$ with members in this coalition, and there would be no other singletons in $C_s$; otherwise, it will not be the smallest partition. Hence, the ratio is no less than $\frac{(y+1)\cdot\mathsf{min}+y\cdot(2\mathsf{min}+h)}{(y+1)\cdot\mathsf{max}} = \left(3-\frac{2}{y+1}\right)\cdot\frac{\mathsf{min}}{\mathsf{max}}+\left(1-\frac{1}{y+1}\right)\cdot\frac{h}{\mathsf{max}}\geq \frac{7\mathsf{min}+2h}{3\mathsf{max}}$. If $x\leq y$, since $C_S^*$ is the minimal set of coalitions to contain these $x$ players, there are at most $x-1$ additional coalitions in $C_S^*$, which, in the worst case, are all pairs. Also, since $C_S$ and $C_S^*$ cannot be further partitioned, the additional $x-1$ players cannot be singletons in $C_S$; they must be at least paired in $C_S$. The remaining number of players in $C_S$ is $y+1+2(x-1)-x-2(x-1)=y-x+1$, which could all be singletons in the worst case. Namely, there are at most $x-1$ other pairs in $C_s^*$ whose members are paired with members in this coalition in $C_s$, and there could be another $y-x+1$ singletons in $C_s$. Hence, the ratio is no less than $\frac{x\cdot\mathsf{min}+(y-x+1)\cdot\mathsf{min}+(x-1)(2\mathsf{min}+h)}{x\cdot\mathsf{max}}=\frac{(2x+y-1)\mathsf{min}+(x-1)h}{x\cdot \mathsf{max}}\geq \frac{(3x-1)\mathsf{min}+(x-1)h}{x\cdot \mathsf{max}} = \left(3-\frac{1}{x}\right)\cdot\frac{\mathsf{min}}{\mathsf{max}}+\left(1-\frac{1}{x}\right)\cdot\frac{h}{\mathsf{max}}\geq \frac{3\mathsf{min}}{\mathsf{max}}$. (iii) If $C_s^*$ only has coalitions of size $2$, then we can let $r_s = 2x$. Notice that we consider the smallest partition, so the only possible structure of $C_s$ is that the players not belonging to $(S_1,\dots,S_p)$ should form a single coalition, denoted by $T^L$, and let $T^R$ be the corresponding part with it before splitting. If $|T^R| > 1$, then $T^L$ gets an additional value of $h$. Hence, we have $\frac{\mathsf{SW}(C_s)}{\mathsf{SW}(C^*_s)} \geq \frac{x\cdot\mathsf{min}+x(\mathsf{min}+h)}{x\cdot\mathsf{max}}=\frac{2\mathsf{min}+h}{\mathsf{max}}$. If $|T^R|\leq 1$, then there is no additional value of $h$. Hence, when $|T^L|\geq 2$, we have $\frac{\mathsf{SW}(C_s)}{\mathsf{SW}(C^*_s)} \geq \frac{x\cdot \mathsf{min}+x\cdot\mathsf{min}+(x-1)\cdot h}{x\cdot \mathsf{max}} = \frac{2\mathsf{min}+\frac{x-1}{x}h}{\mathsf{max}}\geq \frac{2\mathsf{min}+h/2}{\mathsf{max}}$. When $|T^L| = 1$ (then $x=1$) and $|T^R| = 0$, it becomes the first case of $C_k$ and $C_k^*$ where $C_k^*$ contains the union coalition of $C_k$, so that $\frac{\mathsf{SW}(C_s)}{\mathsf{SW}(C^*_s)} \geq \frac{2\mathsf{min}}{2\mathsf{min}+h}$. When $|T^L| = |T^R| = 1$ (then $x=1$), then we recover $T^L\cup T^R$, and move the coalition that contains $T^R$ in the second part back with $C_s^*$. For that coalition, if there are other players, continue to move the coalitions that contain them from the second part back with $C_s$, and so on so forth, until the players in new $C_s$ and $C_s^*$ are the same. Clearly, it will not affect the bound of the remaining second part. For the new $C_s$, there could not be any new singleton. For the singletons in the new $C_s^*$, similar to former cases, they should be as few as possible. Suppose there are $y$ additional players, then the ratio should be at least $\frac{3\mathsf{min}+h+\frac{y-1}{2}\cdot(2\mathsf{min}+h)}{\mathsf{max}+\frac{y-1}{2}\cdot\mathsf{max}+\mathsf{min}} = \frac{\frac{y+1}{2}\cdot(2\mathsf{min}+h)+\mathsf{min}}{\frac{y+1}{2}\cdot\mathsf{max}+\mathsf{min}}\geq \frac{2\mathsf{min}+h}{\mathsf{max}}$ if $y$ is odd; and at least $\frac{3\mathsf{min}+h+\frac{y-2}{2}\cdot(2\mathsf{min}+h)+\mathsf{min}+h}{\mathsf{max}+\frac{y}{2}\cdot\mathsf{max}} = \frac{(y+2)\cdot\mathsf{min}+\frac{y+2}{2}\cdot h}{\frac{y+2}{2}\cdot\mathsf{max}} = \frac{2\mathsf{min}+h}{\mathsf{max}}$ if $y$ is even.
                
                Here, we additionally consider special cases in above condition (iii) with $|T^L|\geq 2$ (then $x\geq 2$) and $|T^R|\leq 1$ (i.e., the case with bound $\frac{2\mathsf{min}+h/2}{\mathsf{max}}$), for $\left\lceil \frac{\mathsf{max}}{\mathsf{min} + h} \right\rceil\leq 2$. When $|T^R| = 0$, since $x\leq \left\lceil \frac{\delta\cdot\mathsf{min}}{\mathsf{min} + h} \right\rceil \leq \left\lceil \frac{\mathsf{max}}{\mathsf{min} + h} \right\rceil\leq 2$ by Lemma~\ref{lem:amchsize}, then $x$ can only be $2$. Since in this case, $\mathsf{max}\leq 2(\mathsf{min}+h)$, then we have $\frac{\mathsf{SW}(C_s)}{\mathsf{SW}(C^*_s)} \geq \frac{4\mathsf{min}+h}{4\mathsf{min}+4h}$. While for $|T^R| = 1$, since $x+1\leq \left\lceil \frac{\delta\cdot\mathsf{min}}{\mathsf{min} + h} \right\rceil \leq \left\lceil \frac{\mathsf{max}}{\mathsf{min} + h} \right\rceil\leq 2$ leads to $x\leq 1$, which contradicts to the current case that $x\geq 2$, it could not happen.
                \item For the second part, suppose there are $r'$ players here. Similar to former cases, the cases when there are singletons among $(S_{p'+1}',\dots,S_{q'}')$ can be neglected. Suppose among $(S^R_{p+1},\dots,S^R_{p+\tilde{p}})$, there are $p_1$ players form singletons and $p_2$ players form non-singletons. Let $x=r'-p_1-p_2$. Hence, we have $\alpha\geq \frac{p_1(\mathsf{min}+h)+\frac{p_2}{2}(2\mathsf{min}+h)+\frac{x}{2}(2\mathsf{min}+h)}{\frac{x+p_1+p_2}{2}\cdot\mathsf{max}} = \frac{2(p_1+p_2+x)\mathsf{min} + (2p_1+p_2+x)h}{(p_1+p_2+x)\mathsf{max}} \geq \frac{2\mathsf{min}+h}{\mathsf{max}}$.
            \end{itemize}
        \end{itemize}
    \end{itemize}
    Taking all cases together, we derive that when $h\leq \mathsf{max} - 2\mathsf{min}$, the ratio $\alpha$ of an AMC-$h$ policy will satisfy
    \[ \alpha\geq \min\left\{ \frac{2\mathsf{min}}{2\mathsf{min}+h}, \frac{3\mathsf{min}}{\mathsf{max}}, \frac{2\mathsf{min}+h}{\mathsf{max}}, \frac{7\mathsf{min}+2h}{3\mathsf{max}}, \frac{2\mathsf{min}+h/2}{\mathsf{max}} \right\},  \]
    and especially, if $\left\lceil \frac{\mathsf{max}}{\mathsf{min} + h} \right\rceil\leq 2$, it will satisfy
    \[ \alpha\geq \min\left\{ \frac{2\mathsf{min}}{2\mathsf{min}+h}, \frac{3\mathsf{min}}{\mathsf{max}}, \frac{2\mathsf{min}+h}{\mathsf{max}}, \frac{7\mathsf{min}+2h}{3\mathsf{max}}, \frac{4\mathsf{min}+h}{4\mathsf{min}+4h} \right\}.  \]

    Together with the case when $h > \mathsf{max} - 2\mathsf{min}$, we now consider the choice of $h$ as
    \[ h = \begin{cases}
        \mathsf{min} & \mathsf{max} < 4\mathsf{min} \text{ (i.e., }v\in V_3 \text{)}; \\
        2\mathsf{min} & \mathsf{max} \geq 4\mathsf{min} \text{ (i.e., }v\in V_{\geq 4} \text{)}.
    \end{cases} \]
    Notice that when $3\mathsf{min}\leq \mathsf{max} < 4\mathsf{min}$, we have $\mathsf{max} - 2\mathsf{min} \geq \mathsf{min} = h$ and $\left\lceil \frac{\mathsf{max}}{\mathsf{min} + h} \right\rceil < \left\lceil \frac{4\mathsf{min}}{\mathsf{min} + \mathsf{min}} \right\rceil = 2$ always be satisfied. Also, when $\mathsf{max} \geq 4\mathsf{min}$, we have we have $\mathsf{max} - 2\mathsf{min} \geq 2\mathsf{min} = h$ always be satisfied. Therefore, by the formula derived above, we have $\alpha\geq \mathsf{min}\left\{ \frac{1}{2}, \frac{3\mathsf{min}}{\mathsf{max}} \right\}$.
\end{proof}

{
\begin{example}
    Consider the same game introduced in Example~\ref{ex:amcbound} and suppose $\mathsf{max} = 3\mathsf{min}$ (i.e., $v\in V_3$). If using an AMC-$h$ policy described above, we will set a threshold as $h=\mathsf{min}$. Then, players from $a_2$ to $a_\delta$ will no longer join $a_1$, as their marginal contribution becomes $\epsilon$ after subtracting $h$. In this case, we can achieve the optimal social welfare.
\end{example}
}

At last, we use an example to see that the bound is tight. Since $\frac{3\mathsf{min}}{\mathsf{max}}$ is an upper bound of all policies, we only need to consider the case where the competitive ratio is  $1/2$. Consider the following case with $\mathsf{min} = 1$ and $\mathsf{max} = 6$, so that $h=2\mathsf{min} = 2$. There are four players with $v(\{a_1\}) = v(\{a_2\}) = v(\{a_3\}) = v(\{a_4\}) = 1$, $v(\{a_1, a_2\}) = 4+\epsilon$, $v(\{a_3,a_4\}) = v(\{a_1, a_4\}) = v(\{a_2, a_3\}) = 2$, and all other coalitions have value of $6$. When the arrival order $\pi = (a_1, a_2, a_3, a_4)$, the greedy players under AMC-$h$ will form a coalition structure as $\{\{a_1,a_2\},\{a_3\},\{a_4\}\}$ with social welfare being $6+\epsilon$. An optimal coalition structure is $\{\{a_1,a_3\},\{a_2,a_4\}\}$ with social welfare being $12$. Then the ratio will be $\frac{6+\epsilon}{12}$ and approaches to $\frac{1}{2}$ when $\epsilon\rightarrow0$.

\section{Non-irrevocable Policies}\label{sec:non-irr}
In this section, we consider non-irrevocable policies. This means, when a new player joins a coalition, the policy not only distributes the additional marginal contribution but also redistributes the value that has already been achieved by previous players. To implement such a process in an online manner, we propose the following scheme.

\begin{itemize}
    \item Once a new coalition $S$ is formed by a player, we establish a \emph{Bank} $B_S$ for the coalition.
    \item For each player $i$ who joins $S$, her instant distributed value consists of two parts, denoted by $\varphi_i(S,\pi_{|S}) = \bar\varphi_i(S,\pi_{|S}) + \tilde\varphi_i(S,\pi_{|S})$, where $\bar\varphi_i$ is the portion distributed to $i$, as in irrevocable policies, while $\tilde\varphi_i$ represents the portion promised to player $i$ by the policy but temporarily stored in the Bank $B_S$. If no additional players join the coalition $S$, $\tilde\varphi_i$ will be distributed to $i$; otherwise, it might be redistributed (becoming larger or smaller) when newcomers arrive.
\end{itemize}

Intuitively, the above scheme reserves a pool of rewards for newcomers, and it can redistribute some of the rewards to previous players according to their contributions to the whole coalition. In practice, the second part can be implemented as an annual bonus or promised options, which could be adjusted after the project is completed.

In the extreme case, if we set the part $\bar\varphi_i$ to always be $0$, then any kind of non-irrevocable policy can be implemented. However, it does not make sense in practical applications, as players would not trust any promised reward stored in the Bank. Hence, we will also consider a variant of greedy players. Normal greedy players will consider the entire value of $\varphi_i(S,\pi_{|S})$ and try to maximize it. A variant, called \emph{pessimistic greedy players}, assumes that a player might consider that she will get nothing from the second part in the worst-case scenario, and therefore, she only considers maximizing the first part of the reward, $\bar\varphi_i$, which cannot be redistributed. Normal greedy players and pessimistic greedy players can be seen as two extreme cases of the model of players' actions. For non-irrevocable policies, we will consider both of these action models.

With non-irrevocable policies, a direct question is whether we can achieve a better competitive ratio than that of irrevocable policies. Before answering this question, recall that, as stated in Proposition~\ref{prop:ir}, any irrevocable policy is IR. This is obviously not true for any non-irrevocable policy due to redistribution of value created. As a basic requirement for stability, it is reasonable to first consider IR non-irrevocable policies. However, unfortunately, if we insist on the constraint of individual rationality, it is impossible to achieve an upper bound higher than $\frac{3\mathsf{min}}{\mathsf{max}}$, regardless of the action model chosen. The following theorem formalizes this limitation.

\begin{theorem}\label{thm:nirrirub}
    When the characteristic function $v\in V_{\geq 3}$, the competitive ratio of an IR distribution policy cannot exceed $\frac{3\mathsf{min}}{\mathsf{max}}$.
\end{theorem}
\begin{proof}
    Consider the following characteristic function $v\in V_{\geq 3}$: 
    \begin{align*}
        & v(\{x\}) = \mathsf{min} \qquad v(\{y\}) = \mathsf{min} \qquad v(\{z\}) = \mathsf{min} \qquad v(\{x,y,z\}) = \mathsf{max} \\
        & v(\{x,y\}) = \mathsf{min} \qquad v(\{x,z\}) = 2\mathsf{min} \qquad v(\{y,z\}) = 2\mathsf{min} 
    \end{align*}
    and $\pi = (x,y,z)$. Then, when player $y$ comes, by individual rationality, she cannot get any positive reward (no matter whether it could be redistributed in the future) by joining with player $x$; otherwise,  player $x$ will violate IR. Hence, she must form a new coalition. This means when player $z$ comes, whichever coalition she decides to join or create a new one, the final social welfare is always $3\mathsf{min}$. However, the best social welfare they can achieve is to form a grand coalition, which is $\mathsf{max}(\geq 3\mathsf{min})$. Therefore, any IR non-irrevocable distribution policy cannot achieve a competitive ratio greater than $\frac{3\mathsf{min}}{\mathsf{max}}$.
\end{proof}

Therefore, if we want a non-irrevocable policy to achieve a better competitive ratio, we need to sacrifice the property of IR. First, we show an upper bound for competitive ratios to assess whether it is worthwhile to sacrifice IR for better near optimal social welfare.

\begin{theorem}\label{thm:nirrub}
    For any non-anticipative distribution policy, when the characteristic function $v\in V_{\geq 3}$, its competitive ratio cannot exceed $\frac{3}{n}$ if $n\leq\frac{\mathsf{max}}{\mathsf{min}}-1$; otherwise, it cannot exceed $\frac{3\mathsf{min}}{\mathsf{max}}$.
\end{theorem}
\begin{proof}
    Consider an arrival order $\pi$ with the first $k$ players $a_1$, $a_2$, $\dots$, $a_k$ such that $v(\{a_i\}) = \mathsf{min}$ and $v(\{a_1,a_2,\dots,a_i\}) = \mathsf{min}+(i-1)\epsilon$ for any $i\leq k$. For any possible $k\geq 2$, if the policy guides the player $a_k$ to form a new coalition, then there might be the case such that $v(\{a_{k+1}\})=\mathsf{min}$, $v(\{a_1,a_2,\dots,a_{k-1},a_{k+1}\}) = v(\{a_k,a_{k+1}\}) = 2\mathsf{min}$, and $v(\{a_1,a_2,\dots,a_{k+1}\})=\mathsf{max}$. This means whichever coalition player $a_{k+1}$ decides to join or create a new one, the final social welfare is always $3\mathsf{min}$ but the optimal one is $\mathsf{max}$. If the policy insists to guide the new player to join the old coalition until some player with different contribution occurs, then consider the case such that $v(\{a_1,a_2,\dots,a_{k+1}\})=\mathsf{max}$, $v(\{a_i,a_{k+i}\})=\mathsf{max}$ for any $1\leq i\leq k$, and $v(S)=|S|\cdot\mathsf{min}$ for any $S\subseteq\{a_{k+2},\dots,a_{2k}\}$. This means no player should join the first coalition after $a_{k+1}$, and no matter what kind of coalition structure is formed by players after $a_{k+1}$, the final social welfare is $\mathsf{max} + (n/2-1)\mathsf{min}$, while the optimal one is $(n/2)\cdot\mathsf{max}$. Therefore, the competitive ratio is at most 
    
    $ \max\left\{ \frac{3\mathsf{min}}{\mathsf{max}}, \frac{2(\mathsf{max}-\mathsf{min})}{n\cdot\mathsf{max}} + \frac{\mathsf{min}}{\mathsf{max}} \right\} \leq \begin{cases}
        \frac{3}{n} & \text{if } n\leq\frac{\mathsf{max}}{\mathsf{min}} -1; \\
        \frac{3\mathsf{min}}{\mathsf{max}} & \text{otherwise.}
    \end{cases} $
\end{proof}

Theorem~\ref{thm:nirrub} implies that even when we sacrifice the property of IR, it cannot improve the worst-case social welfare if $n$ is sufficiently large. Therefore, we now focus on the cases where $n$ is no larger than $\frac{\mathsf{max}-\mathsf{min}}{\mathsf{min}}$. These cases can also be viewed as situations where the maximum value that can be created through cooperation is much larger than the minimum value. To achieve better social welfare, this intuition suggests that we should try to force players to join an existing coalition and wait for a sufficient large value to be created.

\begin{theorem}\label{thm:nirrlb}
    For both action models with greedy and pessimistic greedy players, there exists a non-irrevocable and non-anticipative policy that has a competitive ratio of $\frac{2}{n}$, when $v\in V_{\geq 3}$ and $n>2$. 
\end{theorem}
\begin{proof}
    For both the greedy and pessimistic greedy players action model, we establish policies with parameter $\mu\in(0,1]$ as follows.
    
    \begin{itemize}
        \item For the greedy players model, when a player $i$ joins a coalition $S$,
        \begin{itemize}
            \item if $v(S)\geq \mu\cdot\mathsf{max}$, then all of $i$'s marginal contribution is given to the previous player ($S\neq\emptyset$ since $v(S)\neq 0$) and $\varphi_i = 0$; 
            \item if $v(S)< \mu\cdot\mathsf{max}$, then $\bar\varphi_i(S\cup\{i\})=0$ and $\tilde\varphi_i(S\cup\{i\})=v(S\cup\{i\})$, i.e., all value of the current coalition is allocated to $i$, but all these value could be reallocated in the future.
        \end{itemize}
        \item For the model with pessimistic greedy players, choose a value $\epsilon<\frac{1}{n^2}$. When a player $i$ joins a coalition $S$,
        \begin{itemize}
            \item if $v(S)\geq \mu\cdot\mathsf{max}$, then all of $i$'s marginal contribution is given to the previous player ($S\neq\emptyset$ since $v(S)\neq 0$), and $\varphi_i = 0$;
            \item if $v(S)< \mu\cdot\mathsf{max}$, then $\bar\varphi_i(S\cup\{i\})=\epsilon\cdot|S+1|\cdot\mathsf{min}$ and $\tilde\varphi_i(S\cup\{i\})=B_S$, i.e., the value given to $i$ that cannot be redistributed is $\epsilon\cdot|S+1|\cdot\mathsf{min}$, and all the value in the Bank is promised to $i$ as she can get it when no new player joins $S$.
        \end{itemize}
    \end{itemize}
    We can observe that the coalition structure formed by two policies in two action models is the same: players first form a single coalition until the value reaches $\mu\cdot\mathsf{max}$, then form the next coalition, and so on so forth. Since $v$ is monotone, there must exist a time $t\leq n$ such that $v(\pi_{\preceq t})\geq\mu\cdot\mathsf{max}$. Hence, the worst case happens when the rest of the players can only create a value of $\mathsf{min}$ together, but they can create a value of $\mathsf{max}$ when paired with someone in the first coalition, i.e., $n=2k$, $v(\{a_1,a_{k+1}\}) = \mu\cdot\mathsf{max}$, $v(\{a_{2\leq i\leq k}, a_{k+i}\}) = \mathsf{max}$ and $v(\{a_1,a_2,\dots,a_k\}) = v(\{a_{k+2}, a_{k+3},\dots,a_{2k}\}) = \mathsf{min}$. Then the ratio for both action models is optimized when $\mu = 1$ with $n >2$:
    $ \alpha \geq \frac{\mu\cdot\mathsf{max} + \mathsf{min}}{\left(n/2-1\right)\cdot\mathsf{max} + \mu\cdot\mathsf{max}} \geq_{\text{if }\mu=1} \frac{2\mathsf{max}+\mathsf{min}}{n\cdot\mathsf{max}} \geq \frac{2}{n} $.
\end{proof}

\section{Conclusions and Future Directions}\label{sec:future}
We study a cooperative game where players arrive sequentially and form coalitions to complete tasks, aiming to design an online policy that maximizes social welfare in monotone, bounded games. Our main result establishes a $\frac{3\mathsf{min}}{\mathsf{max}}$ competitive ratio upper bound for irrevocable policies, and proposes a policy achieving competitive ratio of $\min\left\{\frac{1}{2}, \frac{3\mathsf{min}}{\mathsf{max}}\right\}$. We also discuss non-irrevocable policies.

In our model, we assume that agents arrive one at a time and provide the worst-case competitive ratio that maximizes social welfare. This assumption is also adopted in prior literature on online coalition formation games \citep{boehmer2023causes,bullinger2024stability, flammini2021online}. For the general case where two or more agents may arrive simultaneously, most of our analysis remains valid under the assumption that agents cannot observe the arrival of others in the same round. If they choose different coalitions, it just creates a parallel branch of coalitions; if they choose the same one, we can treat them as a single player. However, if they know the information of each other, we may have to consider the equilibrium of their choices each round. We believe it represents an exciting and non-trivial direction for future work.

There are other several future directions worth investigating. For example, for non-irrevocable policies, we have shown two policies for greedy and pessimistic greedy players, respectively, which lead to the same outcome. An intriguing and general question is whether, for any given policy  $\varphi$ and greedy players, there exists another policy $\varphi'$ such that the same outcome can be achieved with pessimistic greedy players, and vice versa; furthermore, even for a mixture of two types of players. If this is true, we would not need to consider how players worry about the future. 
Another significant area for future work could be considering competitive ratios in expected value form, rather than the worst-case form considered in this work.

\vspace{3mm}
\noindent \textbf{Acknowledgments.} This work is partially supported by  JST ERATO Grant Number JPMJER2301 and JSPS KAKENHI Grant Number JP21H04979.

\bibliographystyle{plainnat}
\bibliography{mybibfile}

\end{document}